\documentclass[runningheads]{llncs}

\newtheorem{condition}{Condition}
\usepackage{algorithmic}
\usepackage[ruled,vlined]{algorithm2e}
\usepackage{subfigmat}
\usepackage{graphicx}
\usepackage{amsmath}
\usepackage{caption}
\usepackage{wrapfig}
\usepackage{fancyhdr}
\newcommand\figref[1]{{Figure \ref{fig:#1}}}
\newcommand\tabref[1]{{Table \ref{tab:#1}}}

\begin{document}
\fancyhead{}
\title{Theoretical Analysis on the Efficiency of Interleaved Comparisons}
%
%
\author{Kojiro Iizuka\inst{1,2}\orcidID{0000-0002-3784-2768} \and
Hajime Morita\inst{1}\orcidID{0000-0003-1041-1914} \and
Makoto P. Kato\inst{2}\orcidID{0000-0002-9351-0901}}
%
%
\institute{
Gunosy Inc., Japan \\
\email{iizuka.kojiro@gmail.com}\\
\email{hajime.morita@gunosy.com}\\
\and
University of Tsukuba, Japan \\
\email{mpkato@acm.org}
}

\maketitle              
%
\begin{abstract}
This study presents a theoretical analysis on the efficiency of {\it interleaving},
an efficient online evaluation method for rankings.
Although interleaving has already been applied to production systems,
the source of its high efficiency has not been clarified in the literature.
Therefore, this study presents a theoretical analysis on the efficiency of interleaving methods.
We begin by designing a simple interleaving method similar to ordinary interleaving methods.
Then, we explore a condition under which the interleaving method is more efficient than A/B testing and find that this is the case when users leave the ranking depending on the item's relevance, a typical assumption made in click models.
Finally, we perform experiments based on numerical analysis and user simulation, demonstrating that the theoretical results are consistent with the empirical results.

\end{abstract}

\section{Introduction}
Online evaluation is one of the most important methods for information retrieval and recommender systems~\cite{preexperiment}.
In particular, A/B tests are widely conducted in web services~\cite{grbovic2018airbnb,okura2017yahoo}.
Although A/B testing is easy to implement, it may negatively impact service experiences if the test lasts a long time and the alternative system (B) is not as effective as the current system (A).
In recent years, interleaving methods have been developed to mitigate the negative impacts of online evaluation, as interleaving is experimentally known to be much more efficient than A/B testing~\cite{chapelle2012large,radlinski2008does}.

Although interleaving methods have already been applied to production systems, 
the source of their high efficiency has not been sufficiently studied.
Hence, this study presents a theoretical analysis of the efficiency of interleaving methods.
A precise understanding of what makes interleaving efficient will enable us to distinguish properly the conditions under which interleaving should and should not be employed.
Furthermore, the theoretical analysis could lead to the further development of interleaving methods.

The following example intuitively explains the efficiency of interleaving.
Let us consider the case of evaluating the rankings $A$ and $B$,
where $A$ has higher user satisfaction than $B$.
Interleaving presents an interleaved ranking of $A$ and $B$ to the user
and evaluates these rankings based on the number of clicks given to the items from each ranking.
For example, $A$ is considered superior to $B$ if more items from $A$ in the interleaved ranking are clicked than those from $B$.
In the interleaved ranking, an item from $A$ can be placed ahead of one from $B$ and vice versa.
If the user clicks on an item and is satisfied with the content,
the user may leave the ranking and not click on the other items.
Because we assume items from ranking $A$ are more satisfactory than those from ranking $B$,
items from $B$ have a lesser chance of being clicked when ranking $B$ alone is presented.
Thus, interleaving enables a superior ranking to overtake implicitly the click opportunities of the other ranking, making it easier to observe the difference in the ranking effects on users' click behaviors.

Following this intuition, this study analyzes the efficiency of interleaving from a theoretical perspective.
We begin by designing a simple interleaving method to generalize representative interleaving methods.
Then, we probabilistically model our interleaving method by decomposing an item click into an examination probability and the item's relevance.
By analyzing the model, we show that interleaving has a lower evaluation error rate than A/B testing in cases when users leave the ranking depending on the item's relevance.

We conduct experiments to validate the theoretical analysis, the results of which confirm that the efficiency of interleaving is superior to that of A/B testing when users leave the ranking based on the item's relevance.
The results are consistent with those of the theoretical analysis. 
Therefore, this study theoretically and empirically verifies the source of the efficiency of interleaving.

The contributions of this study can be summarized as follows:
\begin{itemize}
    \item We discuss the nature of the efficiency of interleaving, focusing on the theoretical aspects on which the literature is limited.
    \item We identify the condition under which interleaving is more efficient than A/B testing according to theoretical analysis.
    \item  Our experiments confirm that the theoretical analysis results are consistent with the empirical results.
\end{itemize}

The structure of this paper is as follows.
Section~\ref{relatedworks} describes the related research.
Section~\ref{preliminary} introduces the notations and other preparations.
Section~\ref{analysis} analyzes the efficiency of interleaving.
Section~\ref{numerical} and Section~\ref{experiment} describe the numerical analysis and the user simulation experiments.
Finally, Section~\ref{conclusion} summarizes the study.

\section{Related Works}
\label{relatedworks}

\subsection{User Click Behavior}
Among the various types of user feedback, modeling of click behavior has been extensively studied.
The user behavior model concerning clicks in the information retrieval area is called the \textit{click model}, and it helps simulate user behaviors in the absence of real users or when experiments with real users would interfere with the user experience.
The term click model was first used in the context of the cascade click model introduced by Craswell et al.~\cite{craswell2008experimental}.
Thereafter, basic click models~\cite{chapelle2009dynamic,dupret2008user,guo2009efficient} were developed and are still being improved in different ways.
Our study uses \textit{examination} and \textit{relevance}, as defined in the click model, to model the interleaving method.
The examination represents a user behavior in that the user examines an item on a ranking and relevance concerns how well an item meets the information needs of the user.

\subsection{Online Evaluation}
Online evaluation is a typical method used to perform ranking evaluations.
Among the various methods of online evaluation, A/B testing is widely used and is easy to implement.
It divides users into groups A and B and presents the ranking A and ranking B to each group of users.
As an extension of A/B testing, improved methods exist that reduce the variance of the evaluation~\cite{preexperiment,oosterhuis2020taking,treemodel,casenetflix}.

Another method of online evaluation is interleaving~\cite{brost2016improved,hofmann2011probabilistic,iizuka2021decomposition,joachims2002optimizing,kharitonov2013using,kharitonov2015generalized,radlinski2008does}, which involves evaluating two rankings, whereas multileaving involves evaluating three or more rankings~\cite{brost2016improved,oosterhuis2017ppm}.
In particular, Pairwise Preference Multileaving~(PPM)~\cite{oosterhuis2017ppm} is a multileaving method that has a theoretical guarantee about \textit{fidelity}~\cite{hofmann2013fidelity} that is related to this study.
Some studies suggest that interleaving methods are experimentally known to be 10 to 100 times more efficient than A/B testing~\cite{chapelle2012large,radlinski2008does}.
However, the discussion concerning from where this efficiency comes is limited.
Related to the analysis in this study, Oosterhuis and de Rijke showed that some interleaving methods can cause systematic errors, including bias when compared with the ground truth~\cite{oosterhuis2020taking}.
In this study, we present the cases in which the examination probability of interleaved rankings affects the efficiency of the interleaving method.

\section{Preliminary}
\label{preliminary}
In this section, we first introduce a simple interleaving method for analysis.
We then define the notations to model A/B testing and the interleaving method.
Finally, we define the efficiency.

\subsection{Interleaving Method for Analysis~(IMA)}



Because theoretically analyzing existing interleaving methods is difficult, 
we introduce an interleaving method for analysis~(IMA), which is a simplified version of existing interleaving methods.
In the remainder of this paper, the term \textit{interleaving} will refer to this IMA.
Our method performs for each round $l$ as follows.
Items are probabilistically added according to the flip of a coin.
If the coin is face up, the $l$-th item in ranking $A$ is used as the $l$-th item in the interleaved ranking $I$.
If the coin is face down, the $l$-th item in ranking $B$ is used as the $l$-th item in the interleaved ranking $I$.
A set of items selected from $A$ (or $B$) is denoted by $TeamA$ (or $TeamB$), and this continues until the length of the interleaved ranking reaches the required length, which we assumed to be same for the input ranking length.
In this study, we also assume that no duplication items exist in the ranking $A$ and $B$ to simplify the discussion.

We use almost the same scoring as that in Team Draft Interleaving~(TDI)~\cite{radlinski2008does} to derive a preference between $A$ and $B$ from the observed clicking behavior in $I$.
We let $c_j$ denote the rank of the $j$-th click in the interleaved ranking $I = (i_1, i_2, \ldots)$.
Then, we attribute the clicks to ranking $A$ or $B$ based on which team generated the clicked result.
In particular, for the $u$-th impression, we obtain the score for IMA as follows:
\begin{eqnarray}
\label{imascoring}
score_{I,A}^u &= | \{c_j: i_{c_j} \in TeamA_u\} | / |I| \\ \nonumber
\textrm{~~~and~~~} score_{I,B}^u &= | \{c_j: i_{c_j} \in TeamB_u\} | / |I|,
\end{eqnarray}
where $TeamA_u$ and $TeamB_u$ denote the team attribution of ranking $A$ and ranking $B$, respectively, at the $u$-th impressions and $|I|$ is the length of the interleaved ranking.
Note that the scores $score_{I,A}^u$ and $score_{I,B}^u$ represent the expected values of the click per item in the interleaved ranking, as the scores are divided by the ranking length.
At the end of the evaluation, we score IMA using the formulas 
$score_{I,A}=\sum_{u=1}^n score_{I,A}^u/n$ and $score_{I,B}=\sum_{u=1}^n score_{I,B}^u/n$,
where $n$ is the total number of impressions.
We infer a preference for $A$ if $score_{I,A} > score_{I,B}$, a preference $B$ if $score_{I,A} < score_{I,B}$, and no preference if $score_{I,A} = score_{I,B}$.

An interleaved comparison method is {\it biased}
if, under a random distribution of clicks, it prefers one ranker over another in expectation~\cite{hofmann2011probabilistic}.
Some existing interleaving methods were designed so that the interleaved ranking would not be biased.
For example, Probabilistic Interleaving~(PI)~\cite{hofmann2011probabilistic}, TDI~\cite{radlinski2008does}, and Optimized Interleaving~(OI)~\cite{radlinski2013optimized} were proven unbiased~\cite{oosterhuis2017ppm,hofmann2011probabilistic}.
Our IMA is also unbiased, because every ranker is equally likely to add an item at each location in the ranking; in other words, the interleaved ranking is generated so that the examination probabilities for the input rankings will all be the same.

\subsection{A/B Testing}
In this paper, we re-define an A/B testing method for compatibility with the IMA.
We use Equation (\ref{imascoring}) for each $u$-th impression to score A/B testing, denoted as $score_{AB,A}^u$ and $score_{AB,B}^u$.
At the end of the evaluation, we score A/B testing using 
$score_{AB,A}=\sum_{u=1}^n score_{AB,A}^u/n$ and $score_{AB,B}=\sum_{u=1}^n score_{AB,B}^u/n$,
where $n$ is the total number of impressions, the same as in IMA.
The difference between A/B testing and the IMA is the policy for generating rankings.
All items in the interleaved ranking are selected from either ranking $A$ only or ranking $B$ only in the A/B testing.
In other words, the probability of selecting $A$ or $B$ is fixed to 1 when generating a single interleaved ranking,
or either ranking $A$ or $B$ is presented to the user at an equal probability. 

\subsection{Notation}
First, we introduce the notations of random variables. 
We denote random variables for clicks as $Y \in \{0,1\}$, the examination of the item on the ranking as $O \in \{0,1\}$, and the relevance of the item as $R \in \{0,1\}$.
The value $O=1$ means a user examines the item on the ranking.
This study assumes $Y=O \cdot R$, 
meaning the user clicks on an item only if the item is examined and relevant.
We denote the probability of some random variable $T$ as $P(T)$, the expected value of $T$ as $E(T)$, and the variance in $T$ as $V(T)$.
In addition, the expected value of and variance in the sample mean for $T_i$ of each $i$-th impression over $n$ times are denoted as $E(\bar{T}^n)$ and $V(\bar{T}^n)$, respectively.

Next, we introduce the notations of random variables for the rankings.
In this study, we evaluate two rankings, $A$ and $B$, and we do not distinguish between ranking and items to simplify the notations.
We denote the random variable of clicks in $TeamA$ as $Y_A$ and the random variable of relevance for ranking $A$ as $R_A$.
$O_{AB, A}$ and $O_{I, A}$ are defined as the random variables of the examination of A/B testing and of interleaving, respectively.
Ranking $A$ and ranking $B$ are interchangeable in the above notations.
We also use $\bullet$ to denote ranking $A$ or $B$.
For example, $O_{I, \bullet}$ refers to the random variable for the examination of interleaving when ranking $A$ or $B$.

\subsubsection{\bf Probabilistic Models.}
This study assumes the probabilistic models of A/B testing and interleaving as follows:
\begin{itemize}
\item $Y_{AB,A} = S_{AB,A} \cdot O_{AB,A} \cdot R_A$ holds where
 a random variable for ranking assignment in A/B testing is denoted as $S_{AB,A} \in \{0,1\}$, where $S_{AB,A}=1$ if ranking $A$ is selected via A/B testing. We note that $E(S_{AB,A})=1/2$, as ranking $A$ or $B$ is randomly selected.
\item $Y_{I,A}= S_{I,A} \cdot O_{I,A} \cdot R_A$ holds where
 a random variable for ranking assignment in interleaving is denoted as $S_{I,A} \in \{0,1\}$, where $S_{I,A}=1$ if the item belongs to $TeamA$. We note that $E(S_{I,A})=1/2$, as the item is selected randomly at each position in $I$ from ranking $A$ or $B$.
\end{itemize}

We assume $Y^i_{AB,A}$, a random variable for clicks in the $i$-th impression, follows the Bernoulli distribution $\text{B}(p_{AB,A})$, where $p_{AB,A}=E(Y_{AB,A})$.
Then, $\bar{Y}^n_{AB,A}$ is defined as $\bar{Y}^n_{AB,A}=\frac{1}{n}\sum_{i=1}^n Y^i_{AB,A}$.
In this definition, $\bar{Y}^n_{AB,A}$ follows a binomial distribution, and $\bar{Y}^n_{AB,A}$ can be considered as a random variable that follows normal distribution when $n \rightarrow \infty$.
We note that 
\begin{equation}
E(\bar{Y}^n_{AB,A})=E(\frac{1}{n}\sum_{i=1}^n Y^i_{AB,A})=\frac{1}{n} \sum_{i=1}^n p_{AB,A}
=E(Y_{AB,A}), \label{eq:expectation}
\end{equation}
\begin{equation}
V(\bar{Y}^n_{AB,A}) = V(\frac{1}{n}\sum_{i=1}^n Y^i_{AB,A}) = \frac{1}{n^2} \sum_{i=1}^n p_{AB,A}(1-p_{AB,A}) =\frac{1}{n} V(Y_{AB,A}) \label{eq:variance}
\end{equation}
holds.
The same holds for $\bar{Y}^n_{AB,B}$, $\bar{Y}^n_{I,A}$, and $\bar{Y}^n_{I,B}$. 
In addition,
\begin{equation}
    E(Y_{AB,A}) 
    =E(S_{AB,A}) E(O_A \cdot R_A)
    =E(S_{AB,A}) E(Y_A) \label{eq:abexpect}
\end{equation}
holds, as $Y_{AB,A} = S_{AB,A} \cdot O_{AB,A} \cdot R_A = S_{AB,A} \cdot O_A \cdot R_A$, and $S_{AB,A}$ is independent of $O_A$ and $R_A$.
Finally, note that $E(\bar{Y}^n_{AB,A})=score_{AB,A}$ holds from the definition.
In the above equation, ranking $A$ and ranking $B$ are interchangeable.


\subsection{Definition of Efficiency}
\label{defefficiency}
In this study, efficiency is defined as the level of evaluation error probability, denoted as $P(Error)$.
The efficiency reflects that the error probability is small, given the statistical parameters.
We demonstrate below that interleaving is more efficient; $P(Error_{I}) < P(Error_{AB})$ holds with some conditions.
More formally, the error probability for A/B testing can be defined as follows:
if $E(Y_A) - E(Y_B) > 0$ and ${\bar Y}^n_{AB,\bullet} \sim \mathcal{N}(E({\bar Y}^n_{AB,\bullet}), V({\bar Y}^n_{AB,\bullet}))$ for $n$ is sufficiently large, then:
\begin{align*}
    & P(Error_{AB}) \\ 
    =~& P({\bar Y}^n_{AB,A} - {\bar Y}^n_{AB,B} \le 0) \\
    =~& \int_{-\infty}^{0}\mathcal{N}(x|E({\bar Y}^n_{AB,A}) - E({\bar Y}^n_{AB,B}), V({\bar Y}^n_{AB,A}) + V({\bar Y}^n_{AB,B}))dx  \\
    =~& \int_{-\infty}^{-(E({\bar Y}^n_{AB,A}) - E({\bar Y}^n_{AB,B}))}\mathcal{N}(x|0,V({\bar Y}^n_{AB,A}) + V({\bar Y}^n_{AB,B}))dx. \\
\end{align*}
The second line to the third line uses the reproductive property of a normal distribution.
In the above equation, ranking $A$ and ranking $B$ are interchangeable, and the error probability of interleaving $P(Error_{I})$ is also given in the same way.

\section{Theoretical Analysis}
\label{analysis}
This section discusses the theoretical efficiency of interleaving.
From the definition of efficiency given in Section~\ref{defefficiency}, we see that the A/B testing and interleaving error probability rates depend on the sum of the variances and the difference between the expected click values.
In particular, the smaller the sum of the variances, the smaller the error probability, and the larger the difference between the expected click values, the smaller the error probability.
We investigate the relationship between the variance and difference in the expected click values in the following two cases: when the examination probability is constant and when it is relevance-aware.

\subsubsection{\textbf{Case of Constant Examination.}}
The case of constant examination means the examination probability is constant for the relevance; in other words, the examination probability does not depend on the relevance of a ranking.
For example, a position-based click model~\cite{chuklin2015click} in which the examination probability depends only on the item's position in the ranking is a constant case.
Another example is a perfect click model in the cascade click model used in Section~\ref{experiment}.

We show that interleaving has the same efficiency as A/B testing for the case of constant examination based on two theorems under the following conditions.
\begin{condition}
\label{c0}
$E(O_{AB,A})=E(O_{I,A})=E(O_{AB,B})=E(O_{I,B})=c$: The expected value of the examination is constant and the same between A/B testing and interleaving.
\end{condition}

\begin{condition}
\label{c3}
$O \perp R$: The random variable of the examination $O$ and the random variable of the relevance $R$ are independent from each other.
\end{condition}




\begin{theorem}
\label{staticprob}
If conditions \ref{c0} and \ref{c3} are satisfied, $E({\bar Y}^n_{AB,A}) - E({\bar Y}^n_{AB,B})=E({\bar Y}^n_{I,A}) - E({\bar Y}^n_{I,B})$ holds.
\end{theorem}

\begin{proof}
When $E(O_{AB,A})=E(O_{I,A})$ from condition~\ref{c0}, 
\begin{equation}
E(Y_{AB,A})=E(Y_{I,A})
\label{eqalabandi}
\end{equation}
holds because
$E(Y_{AB,A}) = E(S_{AB,A} \cdot O_{AB,A} \cdot R_A) = E(S_{AB,A}) E(O_{AB,A}) E(R_A) = E(S_{I,A}) E(O_{I,A}) E(R_A) = E(Y_{I,A})$ from condition~\ref{c3} and  $E(S_{AB,A})= E(S_{I,A})=\frac{1}{2}$.
The same holds for ranking $B$, as $E(Y_{AB,B})=E(Y_{I,B})$.
Thus, $E(Y_{AB,A}) - E(Y_{AB,B}) = E(Y_{I,A}) - E(Y_{I,B})$. 
Therefore, $E({\bar Y}^n_{AB,A}) - E({\bar Y}^n_{AB,B})=E({\bar Y}^n_{I,A}) - E({\bar Y}^n_{I,B})$
holds, as $E({\bar Y}^n_{AB,A})=E(Y_{AB,A})$, $E({\bar Y}^n_{AB,B})=E(Y_{AB,B})$, $E({\bar Y}^n_{I,A})=E(Y_{I,A})$, and $E({\bar Y}^n_{I,B})=E(Y_{I,B})$ from Equation \eqref{eq:expectation}.

\end{proof}

\begin{theorem}
\label{staticvariance}
If conditions \ref{c0} and \ref{c3} are satisfied, $V({\bar Y}^n_{AB,A}) + V({\bar Y}^n_{AB,B}) = V({\bar Y}^n_{I,A}) + V({\bar Y}^n_{I,B})$ holds.
\end{theorem}

\begin{proof}
From Equations \eqref{eq:variance} and \eqref{eqalabandi},
$V(\bar{Y}^n_{AB,A})=\frac{1}{n}E(Y_{AB,A})(1-E(Y_{AB,A})) =\frac{1}{n}E(Y_{I,A})(1-E(Y_{I,A})) =V(\bar{Y}^n_{I,A}).$
Similarly, $V(\bar{Y}^n_{AB,B})=V(\bar{Y}^n_{I,B}).$
Thus, $V({\bar Y}^n_{AB,A}) + V({\bar Y}^n_{AB,B}) = V({\bar Y}^n_{I,A}) + V({\bar Y}^n_{I,B})$ holds.

\end{proof}
From Theorems \ref{staticprob} and \ref{staticvariance}, both the difference and variance in the expected click values are the same for interleaving and A/B testing.
Thus, interleaving has the same efficiency as A/B testing when the examination is constant.

\subsubsection{\textbf{Case of Relevance-Aware Examination.}}
In the case of relevance-aware examination, the examination probability depends on the relevance of a ranking.
For example, the navigational click model~\cite{chuklin2015click} is a relevance-aware case in which the click action depends on the relevance, and the click action affects the examination.

We consider the error rate when the examination is relevance-aware under the following conditions.
\begin{condition}
\label{c1}
$E(R_A) > E(R_B) \land E(Y_A) > E(Y_B)$: The expected value of the relevance and the click of ranking $A$ is greater than the expected value of the relevance and the click of ranking $B$.
\end{condition}
\begin{condition}
\label{c2}
$E(O_{I,\bullet} \cdot R_\bullet) \simeq f( max[E(R_A), E(R_B)] ) E(R_\bullet)$,
where $f$ is a monotonically decreasing function.
\end{condition}

The second condition in function $f$ reflects that the user is more likely to leave the ranking after clicking on an item with high relevance.
In particular, a decreasing condition in $f$ means the user is leaving the ranking, and $max$ in $f$ means the leaving behavior is affected more by higher relevance items.
For example, users leave the ranking after interacting with the high relevance items in the navigational click model~\cite{chuklin2015click}.
We show below that interleaving is more efficient than A/B testing when these two conditions hold, based on the following two theorems.




\begin{theorem}
\label{dynamicprob}
If conditions \ref{c1} and \ref{c2} are satisfied, then $E({\bar Y}^n_{I,A}) - E({\bar Y}^n_{I,B}) > E({\bar Y}^n_{AB,A}) - E({\bar Y}^n_{AB,B})$. 
In other words, the difference in the expected click value of the interleaved comparison is greater than the A/B testing value.
\end{theorem}

\begin{proof}
From conditions \ref{c1} and \ref{c2}, 
$E(O_{I,A} \cdot R_A) = f( max[E(R_A), E(R_B)] ) E(R_A) = f(E(R_A)) E(R_A)$.
By interpreting the A/B test as mixing the same input rankings together,  $E(O_{AB,A} \cdot R_A) = f( max[E(R_A), E(R_A)] ) E(R_A) = f(E(R_A)) E(R_A).$

Thus,
\begin{equation}
\label{dynamicequal}
E(Y_{I,A}) = E(Y_{AB,A}),
\end{equation} 
as $E(O_{I,A} \cdot R_A) = E(O_{AB,A} \cdot R_A)$
holds.
Similarly, $E(O_{AB,B} \cdot R_B) = f( E(R_B)) E(R_B)$ holds by interpreting the A/B test as mixing the same input rankings.
In addition, $E(O_{I,B} \cdot R_B) = f( max[E(R_A), E(R_B)] ) E(R_B) = f( E(R_A) ) E(R_B).$
Then, $E(O_{AB,B} \cdot R_B) > E(O_{I,B} \cdot R_B)$ holds because $f$ is a monotonically decreasing function and $E(R_A) > E(R_B)$.
Therefore,
\begin{equation}
\label{dynamicbigger}
E(Y_{AB,B}) > E(Y_{I,B}).
\end{equation}
Furthermore, 
\begin{equation}
E(Y_{AB,A}) > E(Y_{AB,B}),
\label{gthan}
\end{equation}
as $E(Y_{AB,A})=2E(Y_A)>2E(Y_B)=E(Y_{AB,B})$ holds from Equation \eqref{eq:abexpect} and condition \ref{c2}.

From equations (\ref{dynamicequal}), (\ref{dynamicbigger}), and (\ref{gthan}), $E(Y_{I,A}) = E(Y_{AB,A}) > E(Y_{AB,B}) > E(Y_{I,B})$ holds.
Using this relationship, we get $E({\bar Y}^n_{I,A}) - E({\bar Y}^n_{I,B}) > E({\bar Y}^n_{AB,A}) - E({\bar Y}^n_{AB,B}),$
as $E({\bar Y}^n_{AB,A})=E(Y_{AB,A})$, $E({\bar Y}^n_{AB,B})=E(Y_{AB,B})$, $E({\bar Y}^n_{I,A})=E(Y_{I,A})$ and $E({\bar Y}^n_{I,B})=E(Y_{I,B})$ from equation \eqref{eq:expectation}.

\end{proof}

Next, we show that the sum of the variances of interleaving is less than that of A/B testing.
\begin{theorem}
\label{dynamicvariance}
If conditions \ref{c1} and \ref{c2} hold, $V({\bar Y}^n_{AB,A}) + V({\bar Y}^n_{AB,B}) > V({\bar Y}^n_{I,A}) + V({\bar Y}^n_{I,B})$.
\end{theorem}
\begin{proof}
Recall that $V({\bar Y}^n_{AB,\bullet}) = \frac{1}{n} E(Y_{AB,\bullet})(1-E(Y_{AB,\bullet}))$ from Equation \eqref{eq:variance}. We note that $V({\bar Y}^n_{AB,\bullet})$ is monotonically increasing according to $E(Y_{AB,\bullet})$ if the value of $E(Y_{AB,\bullet})$ is less than or equal to $\frac{1}{2}$.
In fact, $E(Y_{AB,\bullet})$ is less than or equal to $0.5$ because $E(Y_{AB,\bullet})=E(S_{AB,\bullet} \cdot O_{AB,\bullet} \cdot R_{\bullet})=E(S_{AB,\bullet})E(O_{AB,\bullet} \cdot R_{\bullet}),$ where $E(S_{AB,\bullet})=\frac{1}{2}$ and $E(O_{AB,\bullet} \cdot R_{\bullet}) \le 1$.
The same holds for $E(Y_{I,\bullet})$.
From equation $E(Y_{I,A}) = E(Y_{AB,A})$ in (\ref{dynamicequal}), we get $V( {\bar Y}^n_{AB,A})=V({\bar Y}^n_{I,A})$.
Furthermore, from equation $E(Y_{AB,B}) > E(Y_{I,B})$ in (\ref{dynamicbigger}), we get $V({\bar Y}^n_{AB,B}) > V({\bar Y}^n_{I,B})$.
Thus,
\[
V({\bar Y}^n_{AB,A}) + V({\bar Y}^n_{AB,B}) > V({\bar Y}^n_{I,A}) + V({\bar Y}^n_{I,B}).
\]
\end{proof}

Based on Theorems \ref{dynamicprob} and \ref{dynamicvariance}, the difference in the expected click values is greater and the variance is lesser in interleaving than in A/B testing.
Thus, interleaving is more efficient than A/B testing when the examination probability depends on the relevance and when the user leaves the ranking according to the relevance.

\section{Numerical Analysis}
\label{numerical}

\begin{figure*}
\begin{tabular}{ccc}

\begin{minipage}{0.33\hsize}
\begin{center}
\includegraphics[width=40mm]{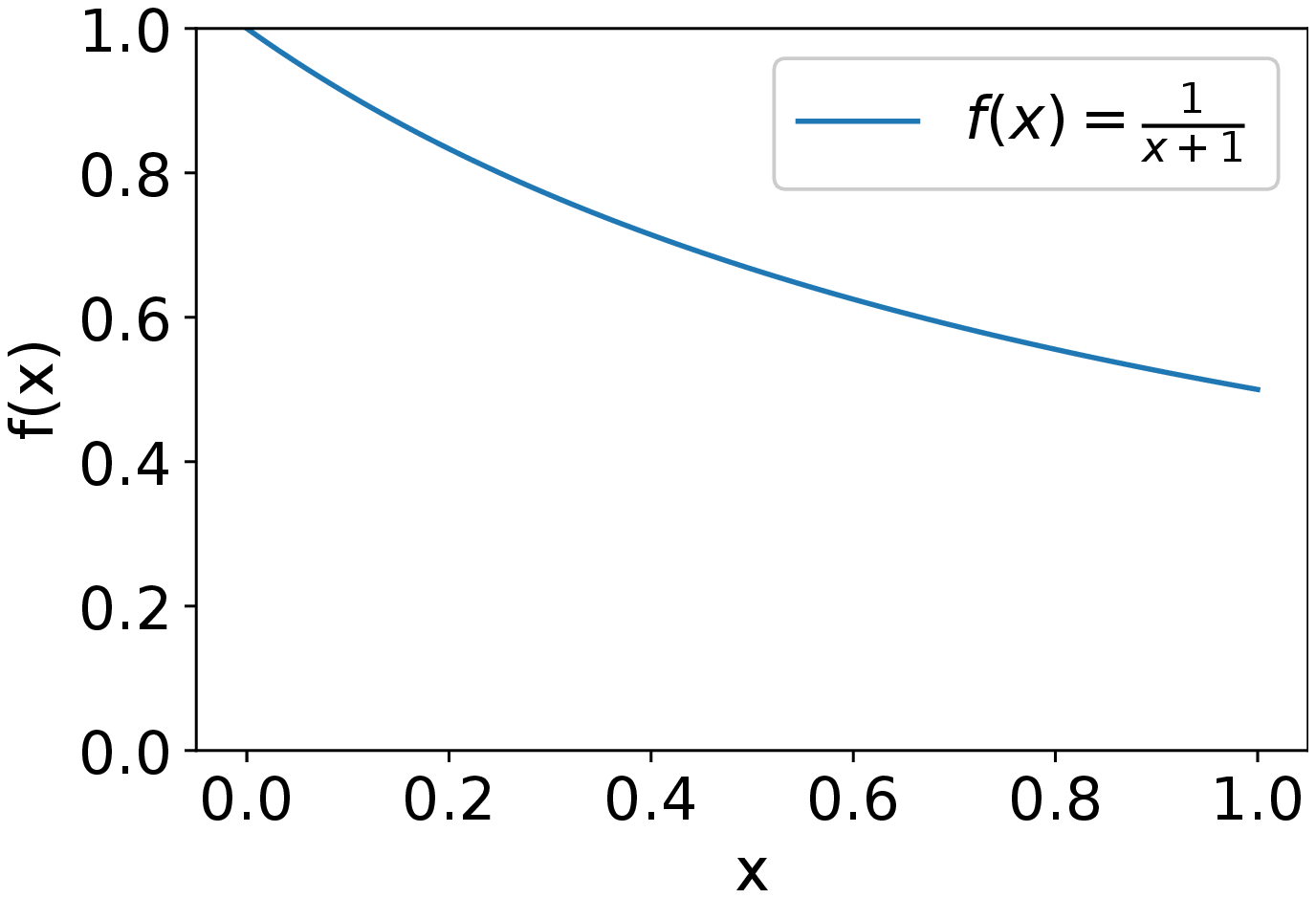}
\end{center}
\end{minipage}

\begin{minipage}{0.33\hsize}
\begin{center}
\includegraphics[width=40mm]{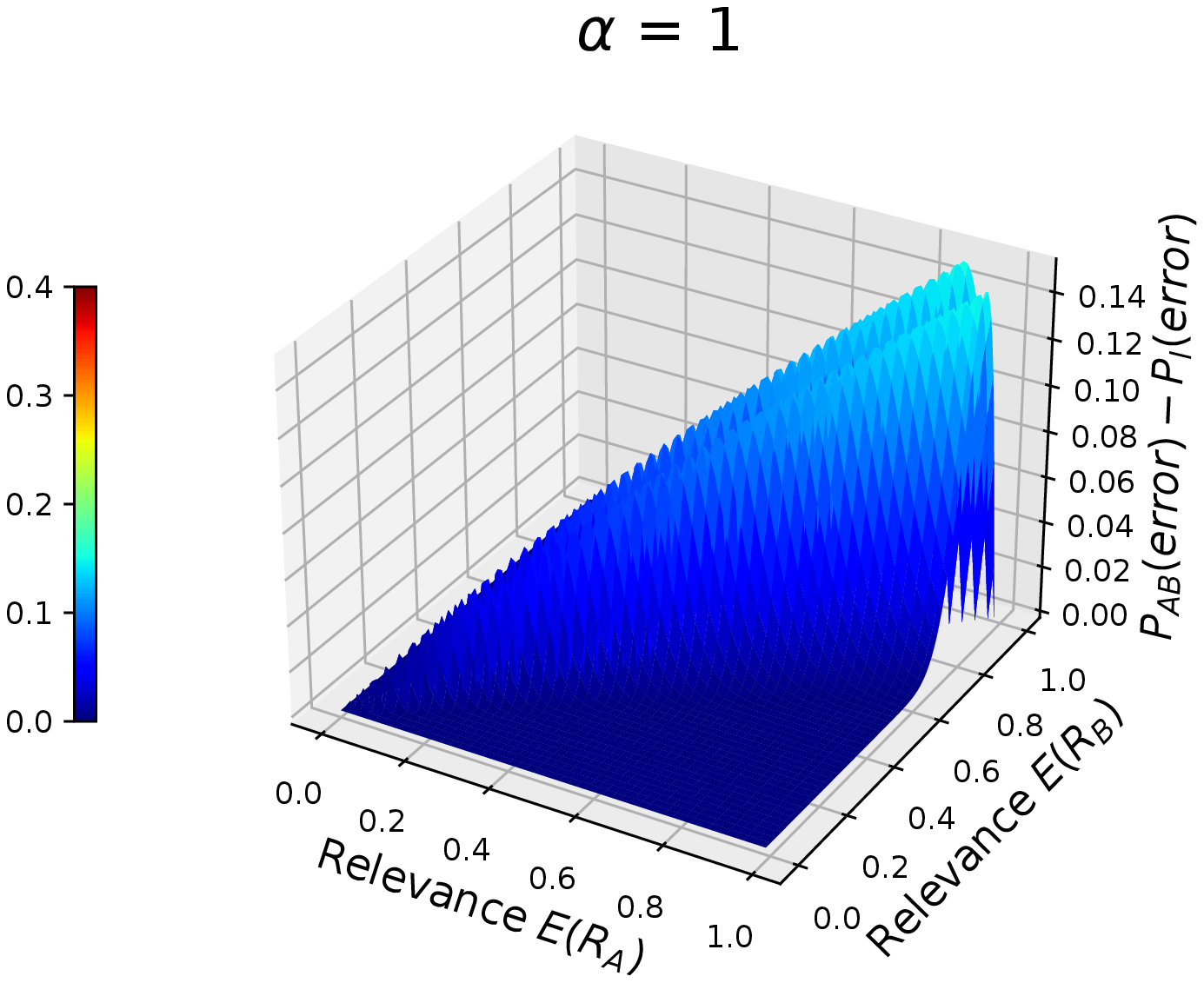}
\end{center}
\end{minipage}

\begin{minipage}{0.33\hsize}
\begin{center}
\includegraphics[width=40mm]{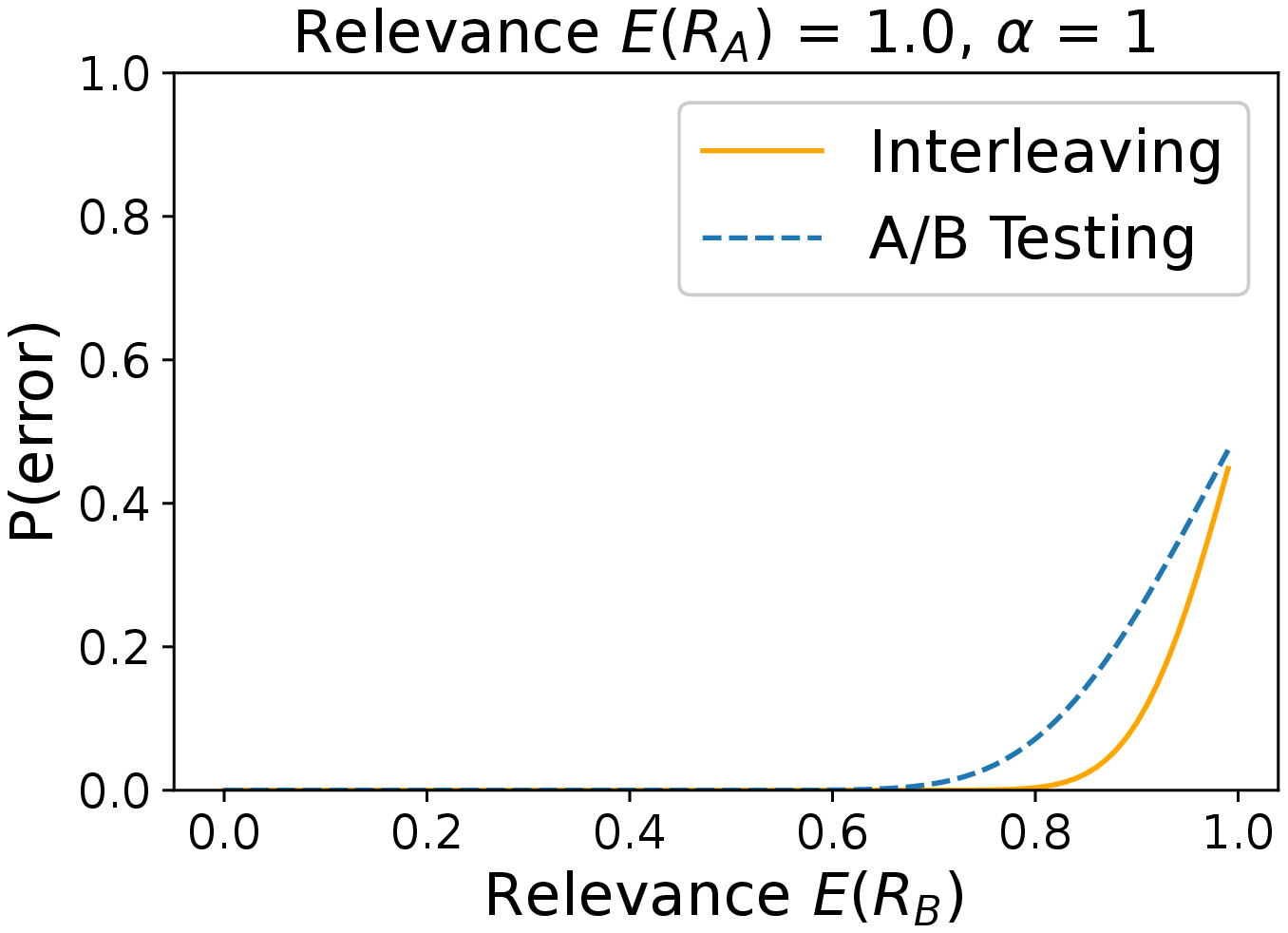}
\end{center}
\end{minipage}

\end{tabular}
\label{fig:logarticleconsumption}






\begin{tabular}{ccc}

\begin{minipage}{0.33\hsize}
\begin{center}
\includegraphics[width=40mm]{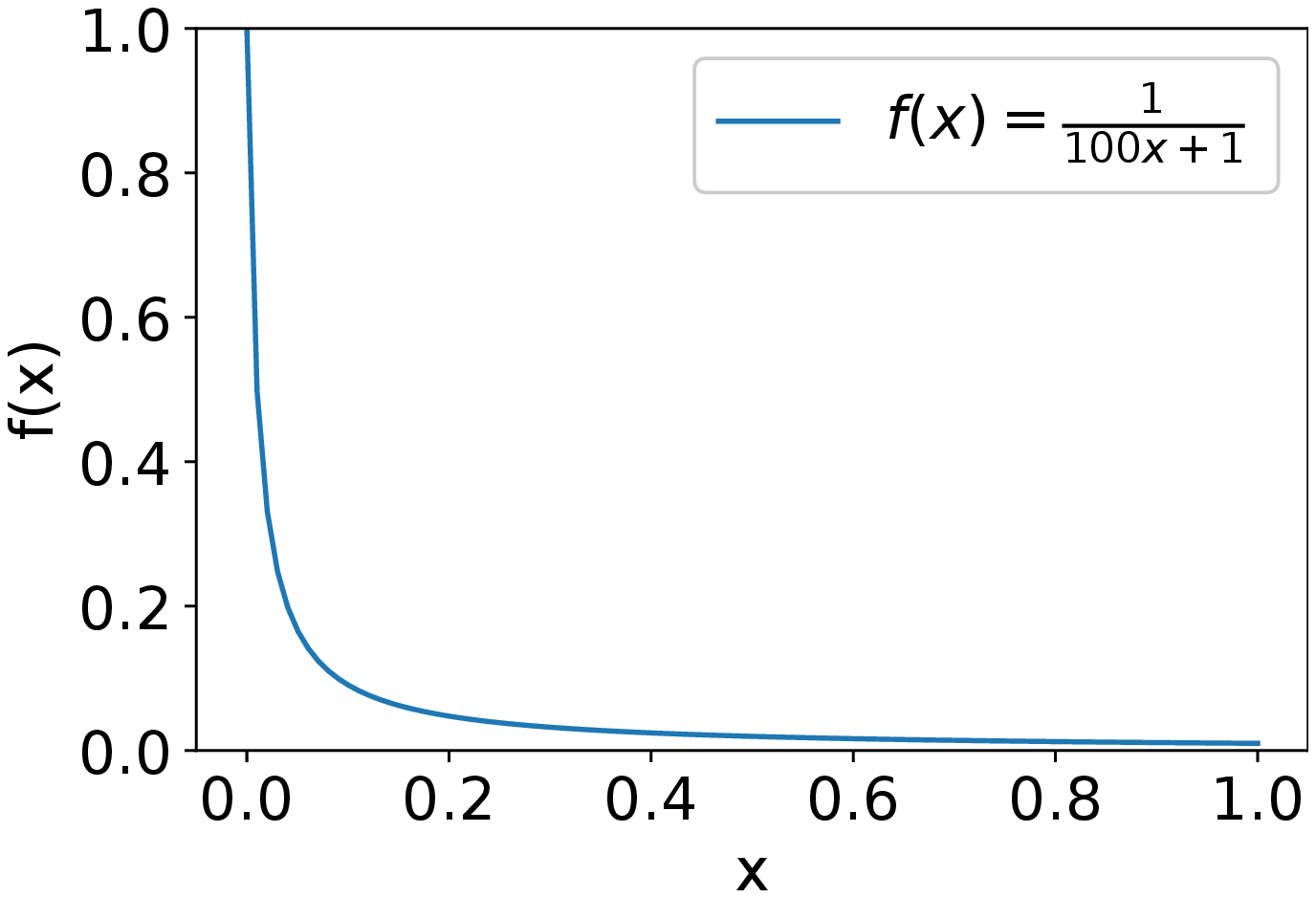}
\captionsetup{width=.9\linewidth}
\caption{Shape of the examination probability $f$}
\label{fig:f}
\end{center}
\end{minipage}

\begin{minipage}{0.33\hsize}
\begin{center}
\includegraphics[width=40mm]{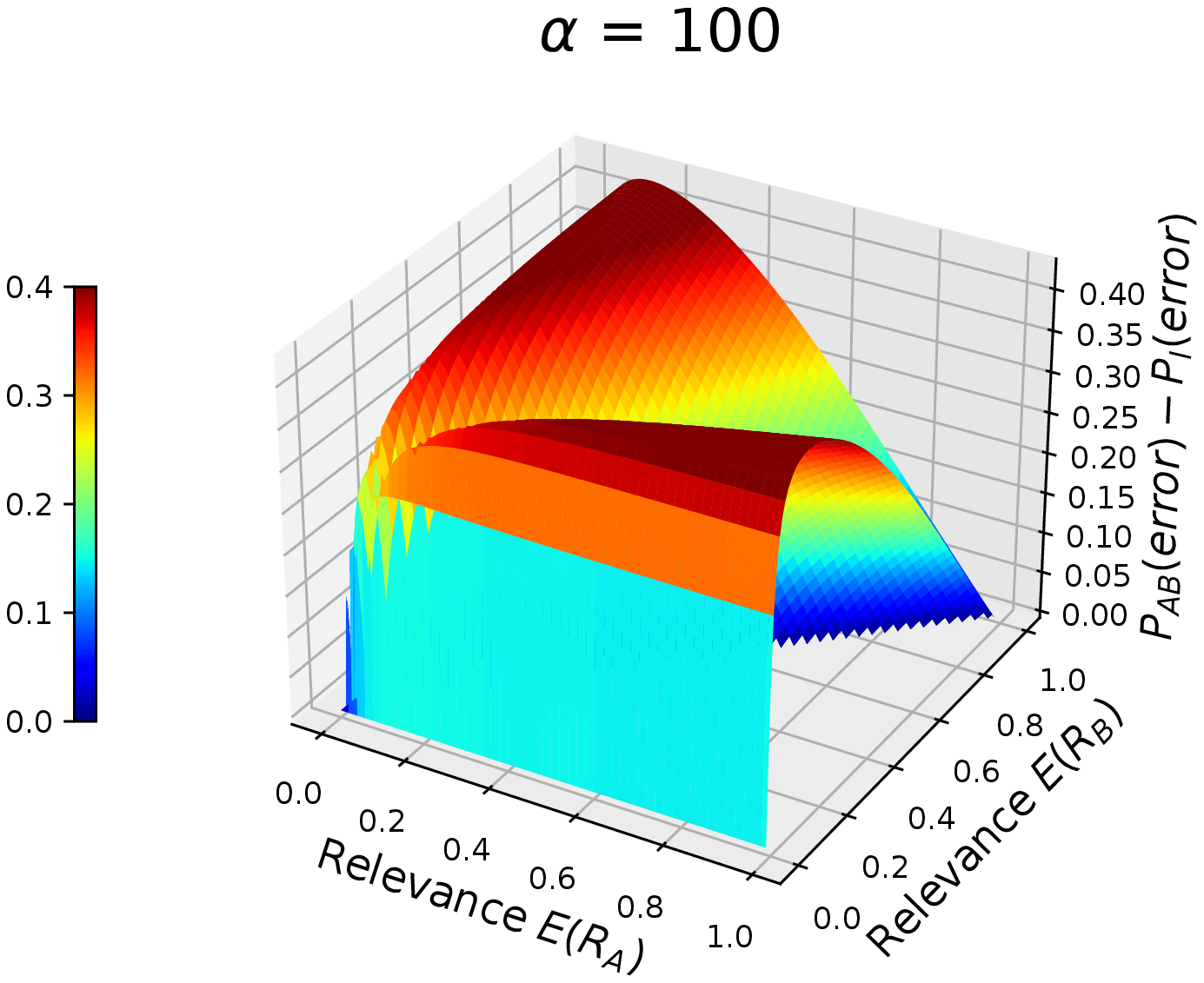}
\captionsetup{width=.9\linewidth}
\caption{Difference in the error probability between A/B testing and interleaving}
\label{fig:diff}
\end{center}
\end{minipage}

\begin{minipage}{0.33\hsize}
\begin{center}
\includegraphics[width=40mm]{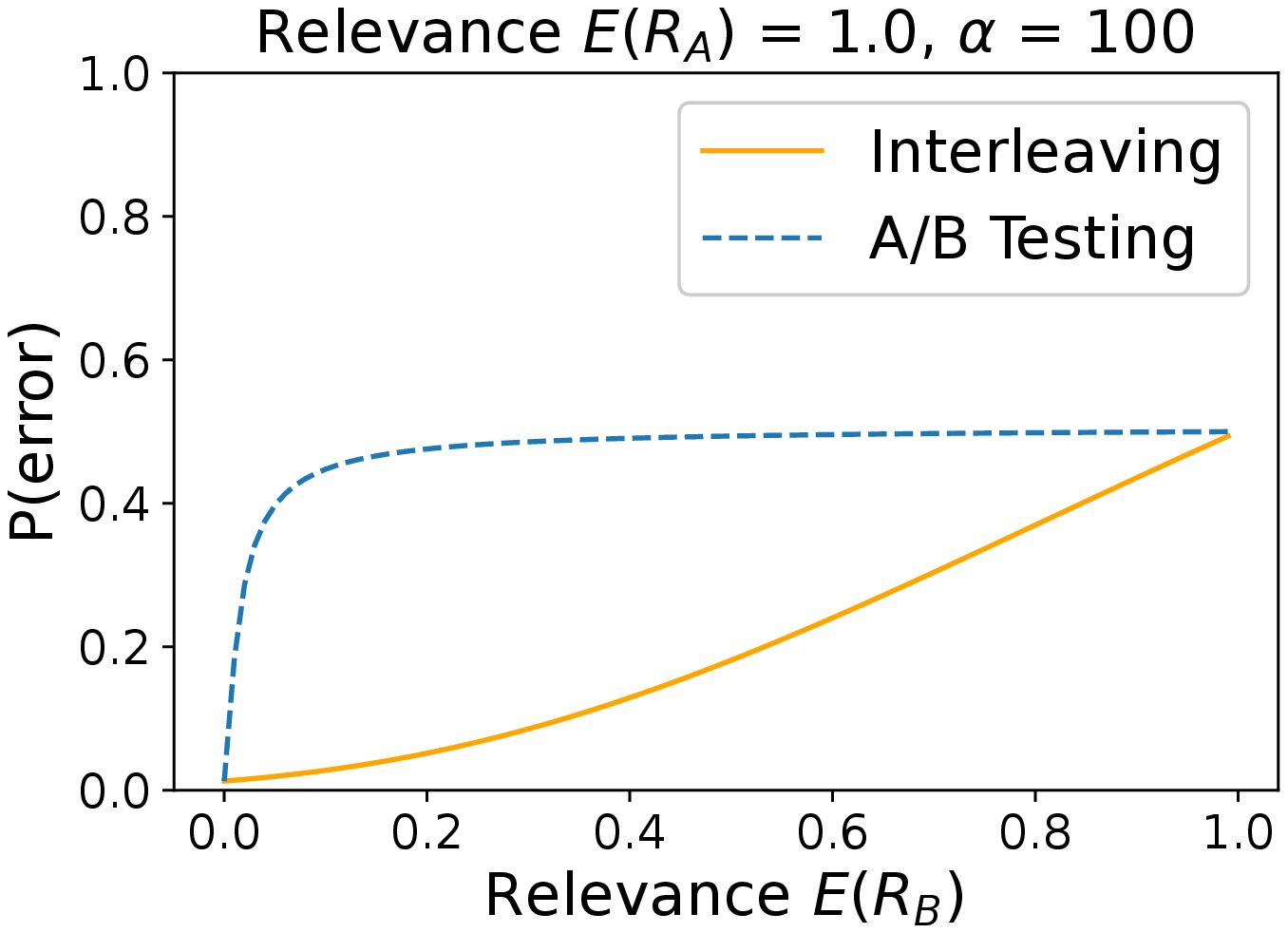}
\captionsetup{width=.9\linewidth}
\caption{Error probability of A/B testing and interleaving}
\label{fig:cut}
\end{center}
\end{minipage}

\end{tabular}
\end{figure*}

This section investigates how the examination probability function $f$ affects the error probability in A/B testing and interleaving.
We set $f(x)=\frac{1}{\alpha x + 1}$, where $\alpha \in \{1, 100\}$ controls how likely a user is to leave the ranking based on the item's relevance.
This function has the following properties that represent actual user behavior:
\begin{itemize}
    \item $f(0)=1$; users must examine the item on the ranking at position $k+1$ when the top-$k$ items are irrelevant;
    \item $f$ is a monotonically decreasing function; users leave the ranking based on its relevance.
\end{itemize}

\figref{f} shows the shape of function $f$, where the x-axis is the level of relevance and the y-axis is the examination probability.
The bottom figure with $\alpha=100$ represents cases when users are most likely to leave the ranking; the top figure with $\alpha=1$ represents cases when they are least likely to do so.

\figref{diff} shows the difference in the error probability between A/B testing and interleaving, that is, $P(Error_{AB})-P(Error_{I})$, where the x-axis is the relevance $E(R_A)$ and the y-axis is the relevance $E(R_B)$.
We observe that when $\alpha$ is larger, the $P(Error_{AB})-P(Error_{I})$ level increases.
\figref{diff} also shows that a greater difference between relevances $E(R_A)$ and $E(R_B)$ indicates a greater difference between $P(Error_{AB})$ and $P(Error_{I})$ when $\alpha=100$, whereas the lesser the difference between relevance $E(R_A)$ and $E(R_B)$, the greater the difference between $P(Error_{AB})$ and $P(Error_{I})$ at $\alpha=1$.
This result implies that interleaving is more efficient than A/B testing if the user is likely to leave the ranking based on relevance and if the difference in the relevance is large.
We further validate these results using user simulations in Section~\ref{experiment}.

\figref{cut} shows the error probability of A/B testing and interleaving, where the relevance $E(R_A)$ is fixed at $1.0$, the x-axis is the relevance $E(R_B)$, and the y-axis is the error probability.
When $\alpha=100$, the error probability of A/B testing is around $0.5$ even if $E(R_B)$ is at $0.2$, whereas the error probability of interleaving is around $0.0$.
\figref{cut} implies that evaluating the difference in the expected click value using A/B testing is difficult if the user is likely to leave the ranking based on a small relevance level.
In contrast, the interleaving method more stably evaluates the difference in the expected click value. 

\section{User Simulation}
\label{experiment}

\begin{table}[t!]
\centering
\caption{Click Models}
\begin{tabular}{l|ccc|ccc} \hline \hline
& \multicolumn{3}{c|}{$P(click = 1 | R)$} & \multicolumn{3}{|c}{$P(stop = 1 | R)$} \\ \hline
$R$ &0&~~1&2&0&~~1&2 \\ \hline
Perfect&0.0&~~0.5&1.0&0.0&~~0.0&0.0 \\
Navigational&0.0&~~0.5&1.0&0.0&~~0.5&1.0 \\
\hline
\end{tabular}
\label{tab:cascade}
\end{table}

In this section, we present the results of our user simulations to answer the following research questions (RQs):
\begin{itemize}
    \item {\bf RQ1:} How does the user click model affect the efficiency of A/B testing and interleaving?
    \item {\bf RQ2:} How does the variability in the relevance of the input rankings affect the error rate?
\end{itemize}

\subsection{Datasets}
We use multiple datasets previously adopted in interleaving research~\cite{oosterhuis2017ppm}.
Most of the datasets are TREC web tracks from $2003$ to $2008$ ~\cite{clarke2009overview,qin2016letor,voorhees2003overview}.
HP2003, HP2004, NP2003, NP2004, TD2003, and TD2004 each have $50$--$150$ queries and $1,000$ items.
Meanwhile, the OHSUMED dataset is based on the query logs of the search engine MEDLINE, an online medical information database, and contains $106$ queries.
MQ2007 and MQ2008 are TREC's million query track ~\cite{allan2007million} datasets, which consist of $1,700$ and $800$ queries, respectively.
The relevance labels are divided into three levels: irrelevant (0), relevant (1), and highly relevant (2).

We generate the input rankings $A$ and $B$ by sorting the items with the features used in past interleaving experiments~\cite{oosterhuis2017ppm,schuth2014multileaved}.
We use the BM25, TF.IDF, TF, IDF, and LMIR.JM features for MQ2007.
For the other datasets, we use the BM25, TF.IDF, LMIR.JM, and hyperlink features.
Note that each feature's value is included in the dataset.
We then generate a pair of input rankings $A$ and $B$ with $|I|=5$ for all pairs of features.
The source code of the ranking generations and user simulations are available in a public GitHub repository.\footnote[1]{https://github.com/mpkato/interleaving}

\subsection{User Behavior}
User behavior is simulated in three steps.
First, the ranking is displayed after the user issues a query.
Next, the user decides whether to click on the items in the ranking.
If the user clicks on an item, they will leave the ranking according to its relevance label.
The details of this user behavior are as follows.

\subsubsection{Ranking Impressions.}
First, the user issues a pseudo-query by uniformly sampling queries from the dataset.
The interleaving model generates an interleaved ranking and displays the ranking to the user.
The IMA was used as the interleaving method in this experiment.
In each ranking impression, up to five items are shown to the user.
After this ranking impression, the user simulates a click using the click model.

\subsubsection{Click Model.}
The cascade click model is used to simulate clicks.
\tabref{cascade} presents the two types of the cascade click model, where $P(click = 1 | R)$ represents how the user clicks the item according to the relevance $R$, and $P(stop = 1 | R)$ represents how the user leaves the ranking according to the relevance after the click.
The perfect click model examines all items in the ranking and clicks on all relevant items.
The navigational click model simulates a user looking for a single relevant item.

\subsection{Results}

\begin{figure}[t]

\begin{minipage}[c]{0.5\hsize}
\centering
\includegraphics[scale=0.3]{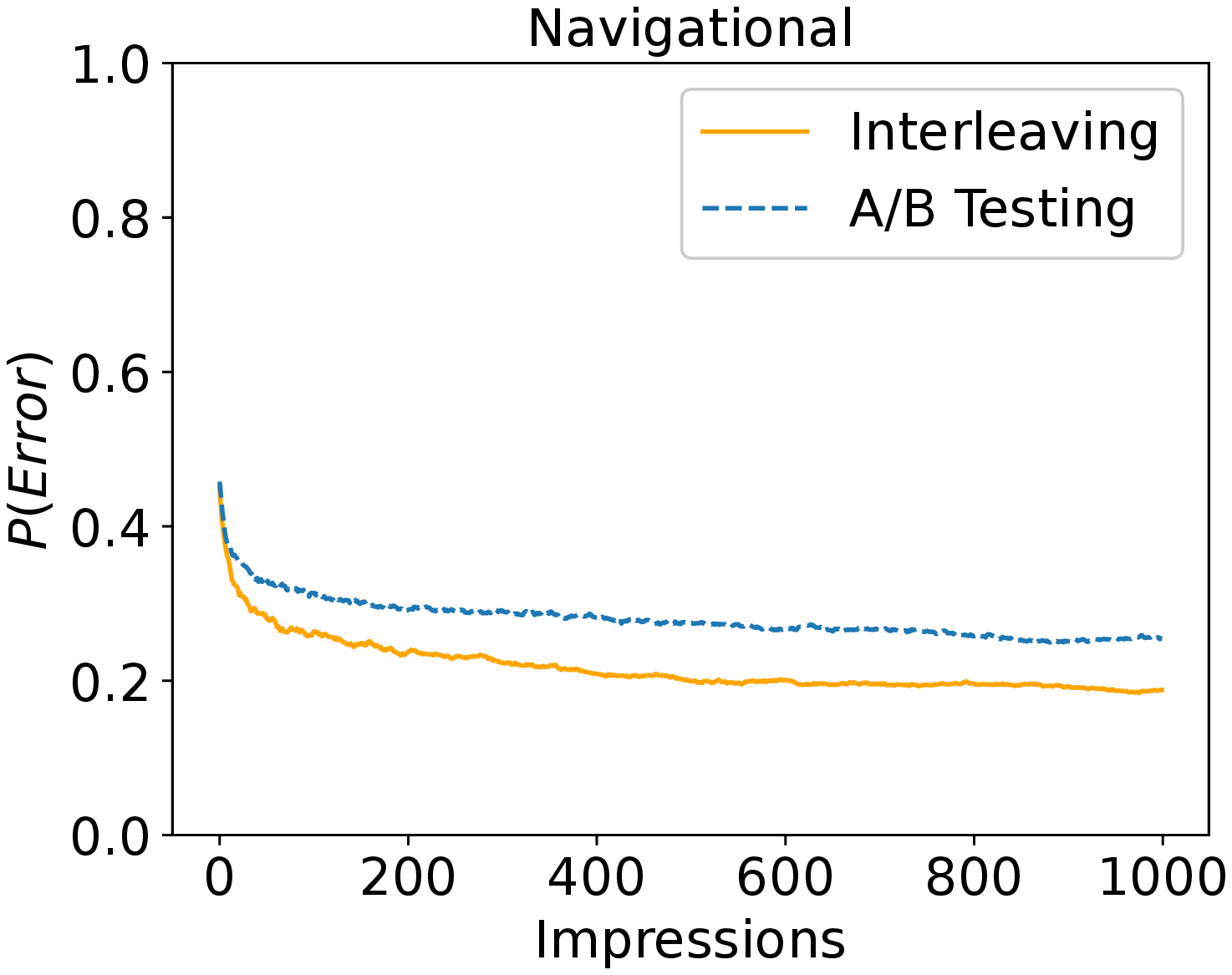}
\end{minipage}
\begin{minipage}[c]{0.5\hsize}
\centering
\includegraphics[scale=0.3]{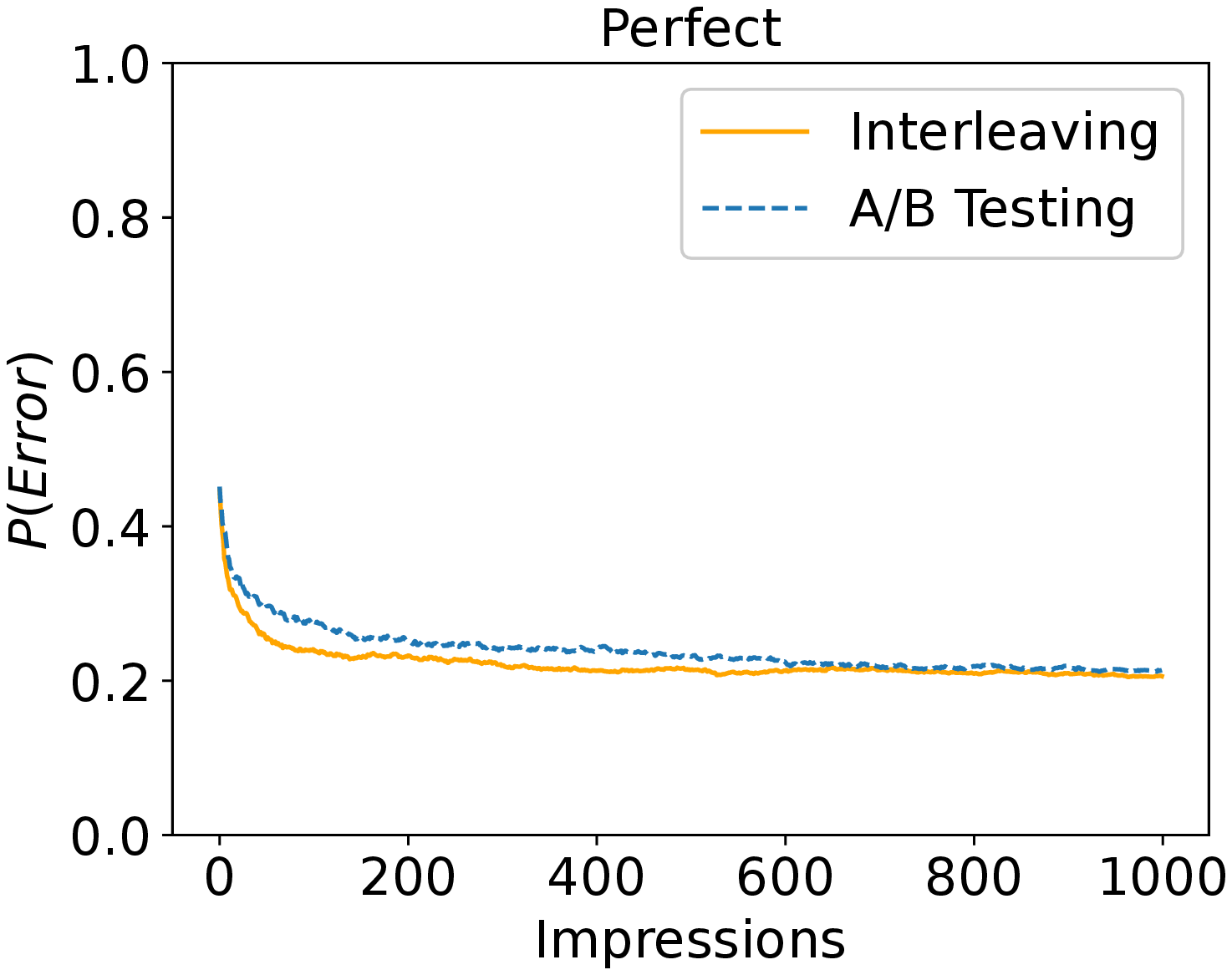}
\end{minipage}
\caption{Efficiency over impressions.}
\label{fig:efficiency}

\begin{minipage}[c]{0.5\hsize}
\centering
\includegraphics[scale=0.3]{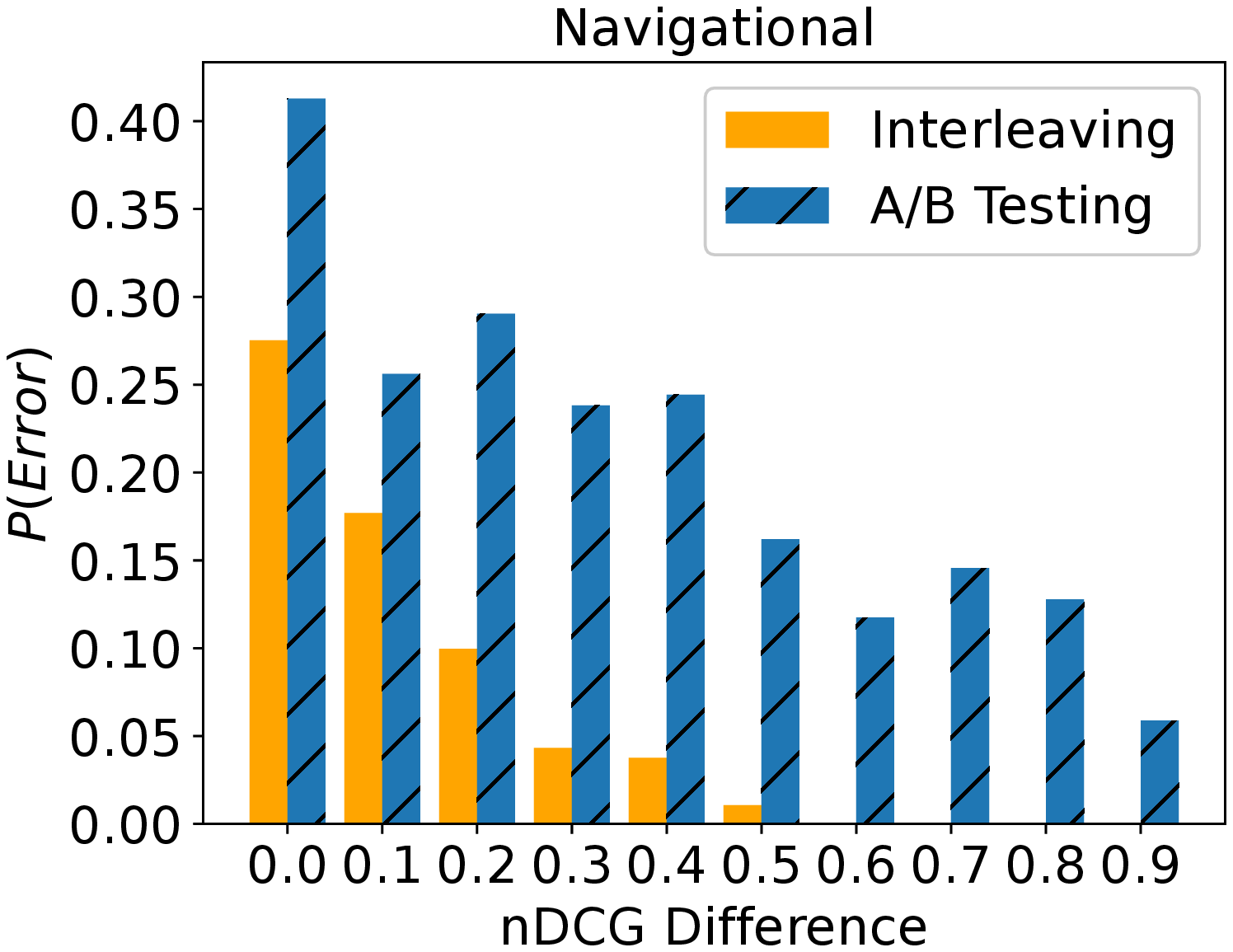}
\end{minipage}
\begin{minipage}[c]{0.5\hsize}
\centering
\includegraphics[scale=0.3]{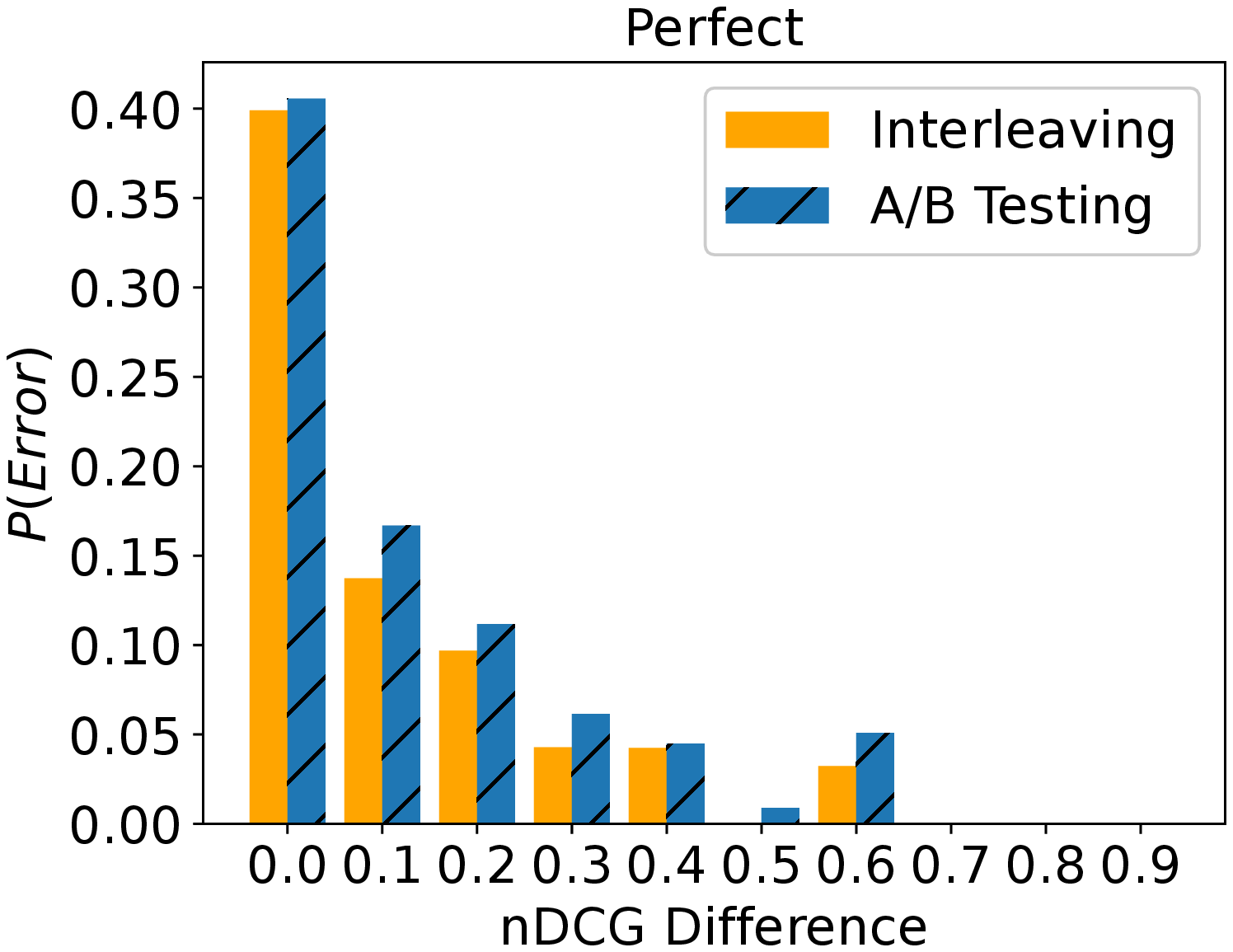}
\end{minipage}
\caption{Error rate for the nDCG difference between rankings.}
\label{fig:ndcgdiff}
 
\end{figure}

\subsubsection{RQ1: How does the user click model affect the efficiency of A/B testing and interleaving?}

\figref{efficiency} presents the error rate over impressions, that is, the efficiency of each click model.
The x-axis represents the number of ranking impressions, and the y-axis represents the error rate.
The values in \figref{efficiency} show the average values.
The maximum number of impressions is $1,000$, and each impression corresponds to a randomly selected query, i.e., one query has one impression.
This evaluation procedure is repeated $10$ times for each dataset.

The results show that interleaving is more efficient and has fewer evaluation errors than A/B testing.
The results are consistent with those of the efficiency analysis, as explained in the numerical experiment showing that interleaving is more efficient than A/B testing if the user is likely to leave the ranking after clicking on the item with the highest relevance.
Thus, the answer to RQ1 is that interleaving is efficient for the navigational click model, where users leave a ranking according to relevance.

\subsubsection{RQ2: How does the variability in the relevance of the input rankings affect the error rate?}

\figref{ndcgdiff} presents the error rate for the nDCG difference between each pair of ranking $A$ and $B$.
The x-axis represents the nDCG difference, and the y-axis represents the error rate.
The values in \figref{ndcgdiff} show the average values for all queries.
We randomly select $1,000$ queries that allow query duplication for each dataset.
The number of impressions is $1,000$ for each query, i.e., one query has $1,000$ impressions, and the evaluation error is calculated after $1,000$ impressions for each query.
This evaluation procedure is repeated $10$ times for each dataset.

The navigational click model shows a lower error rate for interleaving than for A/B testing.
The greater difference between the relevances of ranking $A$ and $B$ corresponds to the greater difference in error rates between interleaving and A/B testing.
These results are consistent with those of our previously illustrated numerical experiment, which shows that interleaving has a lower error rate than A/B testing when the difference in the relevance is large and the user is likely to leave the ranking with relevance,
whereas the perfect click model shows a small or equal difference in the error rate.
From the above analysis, the answer to RQ2 is that based on the navigational click model, interleaving has a low error rate when the ranking pair has a large difference in relevance.


\section{Conclusion}
\label{conclusion}
This study presented a theoretical analysis on the efficiency of interleaving, an online evaluation method, in the information retrieval field.
We first designed a simple interleaving method similar to other interleaving methods, our analysis of which showed that interleaving is more efficient than A/B testing when users leave a ranking according to the item's relevance.
The experiments verified the efficiency of interleaving, and the results according to the theoretical analysis were consistent with the empirical results.

Our theoretical analyses are limited in that they assume no duplication items exist in the $A$ and $B$ rankings.
Item duplication might further contribute to efficiency in some cases beyond our analysis.
However, we believe our basic analysis could be a first step toward discussing the efficiency of more general interleaving settings.
The next challenge includes investigating what other examination probability functions satisfy the efficiency of interleaved or multileaved comparisons.

%
%
%
\bibliographystyle{splncs04}
\bibliography{ref}
%




\end{document}